\documentclass[10pt,conference]{IEEEtran} 

\usepackage{balance}
\usepackage{cite}
\usepackage{hyperref}
\usepackage{xcolor}
\usepackage{multirow}
\usepackage{comment}
\usepackage{braket}
\usepackage{float} 

\usepackage[mathlines,switch]{lineno}

\usepackage{algorithm}
\usepackage{algpseudocode}
\usepackage{placeins}
\usepackage{svg}
\usepackage{amsthm}
\usepackage{graphicx} 
\usepackage{mathtools} 
\usepackage{amsmath} 
\usepackage{amssymb} 

\usepackage{xpatch}

\theoremstyle{definition}
\newtheorem{definition}{Definition}
\newtheorem{cor}{Corollary}

\newtheorem{lemma}{Lemma}
\newtheorem{problem}{Problem}

\theoremstyle{theorem}
\newtheorem{theorem}{Theorem}

\theoremstyle{remark}

\DeclareMathOperator{\E}{\mathbb{E}}

\newcommand\cycle[2][\,]{%
  \readlist\thecycle{#2}%
  (\foreachitem\i\in\thecycle{\ifnum\icnt=1\else#1\fi\i})%
}
\newcommand{\mr}[1]{\mathit{mr}#1}
\newcommand{\hits}[2]{\mathit{hits}_{#2}\left(#1\right)}
\newcommand{\probP}{\text{I\kern-0.15em P}}
\newcommand\bruhatleq{\leq_B}
\newcommand\bruhatlq{<_B}
\newcommand\bruhatlhd{\lhd_B}

\title{Symmetric Locality: Definition and Initial Results}

\author{
    \IEEEauthorblockN{Giordan Escalona
    }
    \IEEEauthorblockA{Department of Computer Science \\
    University of Rochester\\
    Rochester, NY \\
    gescalo2@u.rochester.edu
    }
    \and
    \IEEEauthorblockN{Dylan McKellips
    }
    \IEEEauthorblockA{Department of Computer Science \\
    University of Rochester\\
    Rochester, NY \\
    dmckelli@u.rochester.edu
    }
    \and
    \IEEEauthorblockN{Chen Ding
    }
    \IEEEauthorblockA{Department of Computer Science \\
    University of Rochester\\
    Rochester, NY \\
    cding@cs.rochester.edu
    }
}

\begin{document}

\maketitle
\pagestyle{plain}

\begin{abstract}
In this paper, we characterize \emph{symmetric locality}. In designing algorithms, compilers, and systems, data movement is a common bottleneck in high-performance computation, in which we improve cache and memory performance. We study a special type of data reuse in the form of repeated traversals, or re-traversals, which are based on the symmetric group. The cyclic and sawtooth traces are previously known results in symmetric locality, and in this work, we would like to generalize this result for any re-traversal.  Then, we also provide an abstract framework for applications in compiler design and machine learning models to improve the memory performance of certain programs.
\end{abstract}


\IEEEpeerreviewmaketitle

\section{Introduction}\label{introduction}
\noindent Since data movement is a common bottleneck in high-performance computation, locality is a crucial consideration in designing algorithms, programming, and compiler optimization. In practice, this means improving cache and memory performance.  Hardware caches store recently accessed data blocks, so do operating systems for data pages. Locality of single data elements comes from data reuse.  

\IEEEPARstart{I}{n} this paper, we consider repeated accesses to a collection of data. Each time the data is traversed again is called a \emph{data re-traversal}.  A re-traversal order may be identical, i.e., the same order each time, or different.  This paper presents a theory of the locality of all re-traversal orders.

\IEEEPARstart{T}{wo} types of re-traversal order are common: the cyclic order and the sawtooth order.  In a \emph{cyclic} re-traversal, the order of access is the same as the previous traversal.  
In a \emph{sawtooth} re-traversal, the order is reversed from the previous traversal.  Since cache typically stores recently accessed data, the locality of access corresponds to recency.  A measure of recency is the \emph{reuse distance}, which is the amount of data accessed between two consecutive accesses to the same datum~\cite{relational-locality}.  The reuse distance is first defined by Mattson et al. and called the LRU stack distance~\cite{Mattson+:IBM70}.  

\IEEEPARstart{T}{he} cyclic order has the worst locality in that its reuse distance is the maximal possible.  The cyclic traversal is often referred to as streaming access.  It is widely used as a microbenchmark.  For example, the STREAM benchmark is a standard test to measure the memory bandwidth on a computer~\cite{stream}.  STREAM contains four kernels that each traverse a different number of arrays in the cyclic order.  This benchmark uses it so there is no cache reuse due to poor locality, and data must be transferred repeatedly from memory.

\IEEEPARstart{T}{he} sawtooth order has better locality than cyclic, and more precisely it has the best recency --- the last accessed data is the first to be reused.  Many 
techniques are designed to induce sawtooth data reuse, including the call stack of a running program, insert-at-front heuristic for free lists used by a memory allocator, and the move-to-front heuristic in list search~\cite{SleatorT:CACM85}.

\IEEEPARstart{W}{e} formulate the set of all re-traversal orders, of which cyclic and sawtooth are two members.  The set of re-traversals of $m$ objects corresponds to the set of $m!$ permutations of those objects. The cyclic order corresponds to the identity permutation, and the sawtooth order to the reversed identity permutation.  The set of all permutations of $m$ elements forms a symmetric group $S_m$ as a graded poset (partially ordered set) in algebraic topology~\cite{wachs2006poset}.  We characterize the locality of all possible 
traversal orders and call it the \emph{symmetric locality}.\label{symmlocaldef}

\IEEEPARstart{S}{ymmetric} locality generalizes cyclic and sawtooth orders and encompasses all possible traversal orders.  It may be used in locality optimization: when improving the locality of repeated data traversals, and reordering is restricted, we can choose the one with the best symmetric locality.  In particular, deep learning applications, including the transformer algorithm used in Generative AI systems, repeatedly access the same parameter space.  Symmetric locality can be used to characterize the effect of different traversal orders.

\IEEEPARstart{T}{he} primary contributions of this paper as as follows: 
\begin{itemize}
    \item We characterize locality for symmetric groups, and prove an equivalence between locality ordering and inversion number of permutations
    \item We develop a novel algorithm for calculating reuse distance for traces
    \item We develop a novel algorithm for traversing symmetric groups with optimal locality
\end{itemize}

\IEEEPARstart{T}{he} rest of the paper is organized as follows. 
Section \ref{symmlocal} gives a brief setup of the context and the main problems of the paper. Section \ref{prelim} gives preliminaries in group theory and basic algebraic topology that is needed to understand Section \ref{bruhat}, which contains the theoretical and experimental foundations for our symmetric theory of locality. Section \ref{section:chainfind} defines the chain-finding algorithm and the various total orderings needed for the algorithm. Section \ref{discussion} discusses the applications for symmetric locality, including deep learning with permutation equivariant models and data.

\section{Problem Statement}\label{symmlocal}
\subsection{Outline}
\noindent 
Let $\mathcal{M} = \{1, 2, 3, \hdots, m\}$, where $m$ is the number of trace elements or distinct memory addresses.
A trace $\mathcal{T} = (t_i)_i$ is a sequence of memory addresses such that each $t_i \in \mathcal{M}$.
As a shorthand, we write $\mathcal{T} = A\,B$, where $A = (t_i)_{i=1}^m$ are the first $m$ accesses, and $B = (t_i)_{i = m + 1}^{2m}$ are the next $m$ access after $A$. In particular, we assume that $A$ and $B$ each contain all of $\mathcal{M}$.
\begin{definition}[Re-traversals]
Let $\sigma \in S_m$, where $S_m$ is the symmetric group \ref{symmetric_group} of $m$ elements, and $\sigma$ is a permutation pm $m$ elements as defined in the appendix \ref{appendix}. We construct traces of the form $\mathcal{T} = A\,B$, where
\begin{align*}
    A &= (t_i)_{i=1}^m\;&\;t_i &= i \\
    B &= (t_i)_{i = m + 1}^{2m}\;&\; t_{m+i}  &= \sigma(t_i)
\end{align*}
 This yields $\mathcal{T} = A\,\sigma(A)$, and we can think of $B$ as a reordering of $A$. We call such a trace a periodic trace or a \textbf{re-traversal}.
\end{definition}

For the theory, we assume fully associative caches using least-recently used replacement (LRU).  Modern caches may not use LRU, may be set associative, and often have multiple cache levels.  We consider fully associative LRU cache for two reasons.  First, it is amenable to theoretical analysis in reference to a symbolic cache size $c$.  Second, LRU has the strongest theoretical guarantee in that no online algorithm has better amortized performance than LRU~\cite{SleatorT:CACM85}, and hence most cache policies use some variant of LRU.

\begin{definition}[Miss Ratio]
    For a given $c \in \mathbb{N}$, $\mathit{mr}(c; \mathcal{T})$ denotes the \textbf{miss ratio} at cache size $c$ for trace $\mathcal{T}$. $\mathit{MRC}(\mathcal{T})$ is the miss ratio curve for trace $\mathcal{T}$:
    \[\mathit{MRC}(\mathcal{T}) = \{(c, \mathit{mr}(c; \mathcal{T}))\,:\, \forall c \geq 0\}\]
\end{definition}


\IEEEPARstart{W}{e} are interested in the following problems:

\begin{problem}
\end{problem}
    Let $\mathcal{M} = \{1, 2, 3, \hdots, m\}$ and a cache size $c \leq m$. Let $S_m$ be the symmetric group on $m$ objects, and by observing $\sigma \in S_m$, where $B = \sigma(A)$, what is the miss ratio $\mr(c)$ of $\mathcal{T} = A\,B \;\; \forall c > 0$?
\begin{problem}
    For some program trace $\mathcal{T} = A\,B$, how can we optimize locality by reordering $B$? As in, if $\sigma(A) = B$, is there another $\tau \in S_m$ such that $\tau(A)$ is a valid trace of the program that preserves correctness as well as improves locality?    
\end{problem}
\IEEEPARstart{T}{he} first problem will be the running topic of the paper. The second problem covers applications with regard to machine learning, and will be discussed with more brevity in section \ref{discussion}.

\subsection{Notation}
\noindent Throughout the paper, we are dealing with traces of the form $\mathcal{T} = A\,B$. $A$ can always be understood to be the generic trace of $m$ elements in this fashion: $1, 2, 3, \hdots m$, and if it is not, we can use a relabeling argument. Since $B = \sigma(A)$, and $\sigma$ is a bijective function by construction, we can understand any re-traversal by the permutation $\sigma$ that generates it. Therefore, we will use either characterization interchangeably. Lastly, $A$ is understood to be ordered as in the definition of $\mathcal{M}$, unless explicitly stated.

\section{Preliminaries}\label{prelim}
\subsection{Locality Measures}
\noindent Closely related to $\mr{(c)}$ is the concept of \emph{cache hits}. There exists a hit, if at time $i$ trace element $t_i$ exists in the cache. \begin{definition}
    For a given $\mathcal{T} = A\,B$, $|A| = |B| = m$, let $\hits{\mathcal{T}}{c} \in \{0, 1, \hdots, c\}$ be the number of LRU cache hits that arise after traversing $\mathcal{T}$ with a cache of size $c$.
\end{definition} It may be helpful to refer to $\hits{\mathcal{T}}{C}$, which is the $m$-sized vector: 
\[ \hits{\mathcal{T}}{C} = \left(\hits{\mathcal{T}}{1}, \hits{\mathcal{T}}{2}, \hdots, \hits{\mathcal{T}}{m}\right)\]
For example, let $\mathit{sawtooth}_4$ refer to a well known trace of $m=4$ data items: $a, b, c, d, d, c, b, a$. Upon inspection: \[\hits{\mathit{sawtooth}_4}{C} = (1, 2, 3, 4).\]
The miss ratio for a given $\mathcal{T}$ and $c$ is then given by $\mr{(c; \mathcal{T})} = 1 - \frac{\hits{\mathcal{T}}{c}}{\#\mathit{accesses}}$. 
For the sake of convenience, we will use the $\mathit{hits}$ formula for an equivalent description of locality.

\begin{definition}[Reuse Interval]
    The \textbf{Reuse Interval} of an element in a trace is the number of memory accesses between two accesses of this element. For example, in the trace $abcabc$, the first $a$ has reuse interval 3 as there are three accesses between the first and second $a$, and the second $a$ has reuse distance $\infty$ as there is no third access.
\end{definition}

\begin{definition}[Reuse Distance]
    The \textbf{Reuse Distance} of an element in a trace is the number of unique memory accesses between two accesses of this element. It is equivalent to LRU stack distance. \cite{Mattson+:IBM70} In the previous example, $abcabc$, the reuse distance is the same as reuse interval. However, in the trace $abccba$, the reuse distance of the first access of a would still be 3.
\end{definition}

\subsection{Algebraic Topology}\label{group-theory} 

A review of introductory group theory is covered in the Appendix \ref{groups} In this subsection, we introduce results studied in Algebraic Topology. Primarily, the classification of $S_m$ as a Coxeter group\footnote{While $S_m$'s identity as a Coxeter group implies the existence of the Bruhat order, we do not focus on the Coxeter group, but instead on the partial order.} yields a natural partial order, the \emph{Bruhat order}, defined in section \ref{bruhat}. We show that this Bruhat order coincides with the classifications and orderings of re-traversals based on locality. This provides us a framework with which to analyze orderings.

\subsection{Locality Poset}
\begin{definition}
    Let $\ell(\sigma)$ denote the \emph{minimum} possible length  of the cycle decomposition \ref{cycledecomp} of $\sigma$, ie if $\sigma = \sigma_1 \sigma_2 \sigma_3 \hdots \sigma_n$ where each $\sigma_i$ is an adjacent-disjoint 2-element swap, then $\ell(\sigma) = n$. \\ \emph{Example}: $(13) = (23)(12)(23) \implies \ell(13) = 3$.
\end{definition}

\begin{lemma}
Let $\sigma \in S_m$. Pick $i < j$, if $\sigma(i) > \sigma(j)$, then the pair $(\sigma(i), \sigma(j))$ is an $\textbf{inversion}$ of $\sigma$. Furthermore, $\ell$ is equivalent to the \emph{number of inversions} of $\sigma$:
\[ 
    \ell(\sigma) = \left|\{(i, j)\,:\,\sigma(i) > \sigma(j)\}\right|.\text{\cite{phdthesis}}
\]
\emph{Example}: $\ell (\mathit{sawtooth}_4) = 6$. A simple way to find $\ell(\sigma)$ is to calculate $\sigma(A)$, then count how many $(i, j)$ such that $\sigma(A)[j] > \sigma(A)[i]$ (using 1-indexing) For the trace $\sigma(A) = 2134$, consider the pair (1,2) which has $\sigma(A)[1]>\sigma(A)[2]$, which is an inversion.
\end{lemma}

\begin{lemma}\label{invlemma}
    $\forall \tau \in S_m, \forall \sigma_i \in \mathcal{S}$, the generators of $S_m$,
    \[
        \ell(\tau \sigma_i) = \begin{cases}
            \ell(\tau) + 1 & \tau(i) < \tau(i + 1) \\
            \ell(\tau) - 1 & \tau(i) > \tau(i + 1)
        \end{cases}.\text{\cite{phdthesis}} 
    \]
\end{lemma}
\emph{Example}: Consider $S_5$. Let $\tau = (13)$, and $\sigma_3 = (34)$. Then, $\ell(\tau) = \ell((12)(23)(12)) = 3 \implies \ell(\tau \sigma_3) = 4$.

This lemma also extends to $T \subset S_m$, where $T$ is the set of all simple swaps in $S_m$.

Let $H = (V, E)$ be a digraph (directed graph), such that $V = S_m$ for a given $S_m$. For a given $\sigma, \tau \in V$, we denote $\sigma \bruhatlhd \tau$, if $\exists t \in T$, a \textit{swap} between only two elements such that $\tau = \sigma t$ and $\ell(\sigma) 
 + 1 = \ell(\tau)$. Thus, we let $E$ be the collection of all such relations:
\[ E = \{(\sigma, \tau)\;:\; \sigma \bruhatlhd \tau \}. \]
Since $E$ is generated by $\bruhatlhd$, we use either of these interchangeably.

Let $\lambda$ be an edge labeler defined such that if $Q$ is a totally ordered set,
\[ \lambda : \{(\sigma, \tau)\;:\; \sigma \bruhatlhd \tau\} \to Q. \]
For a given starting point $x \in S_m$, we are interested in traversing or “ascending” the graph by picking the maximal or minimal edges wrt to $\lambda$. We call such paths \textbf{chains}, and the associated greedy algorithm \textit{ChainFind}. This algorithm and $\lambda$ will be discussed later in section \ref{section:chainfind}.

We would also like to note that $H$ has origins from being defined as a \textit{covering graph} using the \textit{inversion} number. In this work, we refer to $H$ as a covering graph interchangeably as the digraph defined here. The \textbf{Bruhat Order} $\bruhatlhd$ arises naturally from its classification as a covering graph. For formal background, refer to Appendix \ref{appendix:covering}. 

\section{Symmetric Theory of Locality}\label{bruhat}
\noindent In this section, we will motivate the usage of Coxeter-Symmetric groups and, more specifically, the Bruhat order to analyze how locality changes w.r.t. the symmetric group $S_m$. We show that the Bruhat order coincides with locality through theoretical proof and experimental results, and include additional miscellaneous results in \ref{misc}. 


\subsection{Inversions}
Section \ref{group-theory} gives a definition of $\ell(\sigma)$, the inversion number of $\sigma \in S_m$. This quantity has been utilized to study efficiency of sorting algorithms. \cite{inversions} The inversion number corresponds to lengths of chains on the covering graph formed by the Bruhat order.

\subsection{Identity} 

By construction, the identity permutation, or the cyclic trace is considered the smallest element in $(S_m, \bruhatlhd)$. More formally, $\ell(e) = 0,$ where $e$ is the identity permutation. This trace, also known as the cyclic trace, has the worst locality, as well.

\subsection{Reverse Identity} $S_m$ is a poset defined by the Bruhat order, and in particular it is a graded poset. In a graded poset, its maximal possible chains from the lowest ordered element to the highest are all the same length, denoted $\ell(S_m)$. The sawtooth trace, which also refers to the reverse identity permutation, has the best locality and has the maximal $\ell$, by the Bruhat order.
\subsection{Locality}

The foundation of the connection between the Bruhat Order imposed on $S_m$ and the locality of a re-traversal is formally stated in the theorems below. 

\begin{theorem}[Reuse Distance Calculation]
We provide an algorithm for calculating the cache hit vector iteratively, which also connects the reuse distance of an element with the inversion number of an element. 
\end{theorem}
\IEEEPARstart{W}{e} are caching with LRU stack distance, also known as reuse distance. 
The cache hit vector corresponds directly to reuse distance; $\mathit{hits}_c(\sigma)$ is exactly the number of elements with a reuse distance of $c$ or smaller. This is a direct result of using LRU caching for counting cache hits, as reuse distance is also known as LRU stack distance. \\We calculate the reuse interval of an element $a \in A$, and subtract the number of repeated elements between each access of $a$.  Denote the rank $r(a)$ of $a$ to be $n-a + 1$. Then, we see the reuse interval is  $r(a)-1 + i$, where $i$ is the index of $a \in \sigma(A)$. The highest rank will be 1 in order for this to correspond with integers $1\dots n$, and similarly we will index arrays starting at 1. \\
\IEEEPARstart{I}{n} order to convert this to reuse distance, we must subtract the number of repeated values. To do this we keep track of a binary vector c, which flips a bit at place $r$ if the element at rank $r$ is accessed. Then, at the time we access element $a$, the sum $\sum_{i = 1}^{r(a) - 1} c[i]$ will calculate the number of repeats we have seen (this also counts the inversion number induced by $a$).  Taking this together, we have our formula for reuse distance $r-1+1 - \sum_{i = 1}^{r(a) - 1} c[i]$. We increment the reuse distance histogram (rdh) and cache hit vector (chv) at this index by 1. Since a cache hit at size $i-1$ will also be a cache hit at size $i$, we add $chv[i-1]$ to $chv[i]$. 
\begin{algorithm}[!h]
    \caption{Reuse Distance Histogram}
    \begin{algorithmic}
        \For{$k \in \sigma(A)$}
            \State $r \gets n - k + 1$
            \State $c[r] \gets 1$
            \State  $repeats \gets \sum_{i=0}^{r-1} c[i]$
            \State $rdh[r-1 + i - repeats] \gets + 1$
            \State $chv[r - 1 + i - repeats] \gets +1$
            \State $chv[i] \gets +chv[i-1]$
            \EndFor
    \end{algorithmic}
\end{algorithm}

We provide an example below:
\begin{alignat*}{9}
\intertext{Denote that we have seen 2, the element of rank 3}
    4 & \quad & 3 & \quad & 2 & \quad & 1 & \quad & 3 & \quad & 4 & \quad & 2 & \quad & 1 & \quad & r \\
    1 & \quad & 2 & \quad & 3 & \quad & 4 & \quad & 2 & \quad & 1 & \quad & 3 & \quad & 4 & \quad & A \, \sigma(A) \\
      & \quad &   & \quad &   & \quad &   & \quad & 0 & \quad & 0 & \quad & 1 & \quad & 0 & \quad & c \\
      & \quad &   & \quad &   & \quad &   & \quad & 0 & \quad & 0 & \quad & 0 & \quad & 0 & \quad & \text{RD Histogram} \\
      & \quad &   & \quad &   & \quad &   & \quad & 1 & \quad & 2 & \quad & 3 & \quad & 4 & \quad & \text{Index} \\
      & \quad &   & \quad &   & \quad &   & \quad & 0 & \quad & 0 & \quad & 0 & \quad & 0 & \quad & \text{Cache Hit Vector}\\
      \intertext{Increment rdh and chv at index 3 - 1 + 1 - 0 = 3}
    4 & \quad & 3 & \quad & 2 & \quad & 1 & \quad & 3 & \quad & 4 & \quad & 2 & \quad & 1 & \quad & r \\
1 & \quad & 2 & \quad & 3 & \quad & 4 & \quad & 2 & \quad & 1 & \quad & 3 & \quad & 4 & \quad & A \, \sigma(A) \\
      & \quad &   & \quad &   & \quad &   & \quad & 0 & \quad & 0 & \quad & 1 & \quad & 0 & \quad & c \\
      & \quad &   & \quad &   & \quad &   & \quad & 0 & \quad & 0 & \quad & 1 & \quad & 0 & \quad & \text{RD Histogram} \\
      & \quad &   & \quad &   & \quad &   & \quad & 1 & \quad & 2 & \quad & 3 & \quad & 4 & \quad & \text{Index} \\
      & \quad &   & \quad &   & \quad &   & \quad & 0 & \quad & 0 & \quad & 1 & \quad & 1 & \quad & \text{Cache Hit Vector}
      \intertext{Denote that we have seen 1, the element of rank 4}
4 & \quad & 3 & \quad & 2 & \quad & 1 & \quad & 3 & \quad & 4 & \quad & 2 & \quad & 1 & \quad & r \\
1 & \quad & 2 & \quad & 3 & \quad & 4 & \quad & 2 & \quad & 1 & \quad & 3 & \quad & 4 & \quad & A \, \sigma(A) \\
      & \quad &   & \quad &   & \quad &   & \quad & 0 & \quad & 0 & \quad & 1 & \quad & 1 & \quad & c \\
      & \quad &   & \quad &   & \quad &   & \quad & 0 & \quad & 0 & \quad & 1 & \quad & 0 & \quad & \text{RD Histogram} \\
      & \quad &   & \quad &   & \quad &   & \quad & 1 & \quad & 2 & \quad & 3 & \quad & 4 & \quad & \text{Index} \\
      & \quad &   & \quad &   & \quad &   & \quad & 0 & \quad & 0 & \quad & 1 & \quad & 1 & \quad & \text{Cache Hit Vector}\\
      \intertext{Increment chv and rd at index 4 - 1 + 2 - 1 , and sum the previous term from the chv}
    4 & \quad & 3 & \quad & 2 & \quad & 1 & \quad & 3 & \quad & 4 & \quad & 2 & \quad & 1 & \quad & r \\
1 & \quad & 2 & \quad & 3 & \quad & 4 & \quad & 2 & \quad & 1 & \quad & 3 & \quad & 4 & \quad & A \, \sigma(A) \\
      & \quad &   & \quad &   & \quad &   & \quad & 0 & \quad & 0 & \quad & 1 & \quad & 1 & \quad & c \\
      & \quad &   & \quad &   & \quad &   & \quad & 0 & \quad & 0 & \quad & 1 & \quad & 1 & \quad & \text{RD Histogram} \\
      & \quad &   & \quad &   & \quad &   & \quad & 1 & \quad & 2 & \quad & 3 & \quad & 4 & \quad & \text{Index} \\
      & \quad &   & \quad &   & \quad &   & \quad & 0 & \quad & 0 & \quad & 1 & \quad & 2 & \quad & \text{Cache Hit Vector}
\end{alignat*}
    
\begin{theorem}[Bruhat-Locality]\label{bruhat-locality}
    Let $S_m$ be the symmetric group of $m$ elements. For $\sigma \in S_m$ and some $C \leq m$, we can compute $\mathit{hits}_C(\sigma) = \mathit{hits}_C(A\,B)$, then
        \[ 
            \sum_{c = 1}^{m - 1}\mathit{hits}_c(\sigma) = \ell(\sigma).
        \] 
\end{theorem}
\begin{proof}
    We prove this by considering the algorithm above, which connects the reuse distance histogram (and thus the cache hit vector) to the inversion number. We consider all variables in the equation $r - 1 + i - \sum_{i=1}^{r(a)-1} c[i]$ that determines the reuse distance of an element. In a swap $(a\,b)$, the rank of each element is unchanged, and the index of one is increased and the index of the other is decreased by an equal amount, so these values do not affect the reuse distance histogram. 
    
    Since we have imposed the Bruhat order, the swap is guaranteed to increase the  inversion number of $\sigma(A)$ by exactly 1, and since $\sum_{i=1}^{r(a)-1} c[i]$ measures the inversion number of each element, the net increase of all sums across 1, resulting in a decrease by 1 of the index of an element. This improves the total reuse distribution (and the sum of cache hit vector) by 1.

\end{proof}
\begin{cor}
    The below formula is equivalent to \ref{bruhat-locality} 
    \[\sum_{c = 1}^{m}\mathit{hits}_c(\sigma) = m + \ell(\sigma).\]
\end{cor}

After reading our initial proof, a colleague Donovan Snyder provided an alternative proof. We include an adaption of this this below.
\begin{proof}[Snyder Proof]
For a permutation $\sigma$, define
\begin{equation*}
    \ell_{a}(\sigma) = |\{j: j>a, \sigma^{-1}(j)<\sigma^{-1}(a)\}|
\end{equation*}
 We can think of this as the number of items that come before $a$ that do not belong in the original permutation. Then, to count the number of elements with a reuse distance of c or less, we define the following: 
 \begin{equation*}
     h_c(\sigma) = | \{a : rd_{\sigma}(a) \leq c \} = \sum_{k=1}^c |rd_\sigma^{-1}(k)|
 \end{equation*} We then have \begin{align*}
    \sum_{c=1}^m h_c(\sigma) &= \sum_{c=1}^m\sum_{k=1}^c |rd_\sigma^{-1}(k)|\\
    &= \sum_{i=1}^m (m+1-i) |rd_{\sigma}^{-1}(i)|\\
    &= (m+1) \sum_{i=1}^m |rd_\sigma^{-1}(i)| - \sum_{i=1}^m |rd_\sigma^{-1}(i)|\\
    &= (m+1)m  - \sum_{i=1}^m rd_\sigma(i)\\
    &= m^2+m - \sum_{i=1}^n (m-i) + \sigma^{-1}(i) - \ell_i(\sigma)\\
    &= m^2 + m - \bigg( \frac{m(m-1)}{2} + \frac{(m+1)m}{2} - \ell(\sigma) \bigg)\\
    &= m^2+m - (m^2 - \ell(\sigma))\\
    &= \ell(\sigma)+m
\end{align*}
    
\end{proof}
\begin{theorem}
If $\sigma \bruhatlhd \tau$, $\mathit{mr}(c; \sigma) \leq \mathit{mr}(c; \tau), \forall c \leq m$.
\end{theorem}
\begin{proof}
    By \ref{bruhat-locality} and definition of $\lhd$:
    \begin{align*}
        &\sum_{c = 1}^{m - 1}\mathit{hits}_c(\tau) = \bigg( \sum_{c = 1}^{m - 1}\mathit{hits}_c(\sigma) \bigg) + 1 \\
        &\sum_{c = 1}^{m - 1}\left(\mathit{hits}_c(\tau) - \mathit{hits}_c(\sigma)\right) = 1
    \end{align*}
    This implies $\exists! c' < m$ such that:
    \[\mathit{hits}_{c'}(\tau) = \mathit{hits}_{c'}(\sigma) + 1\]
    Therefore, $c \neq c'$, $\mathit{mr}(c;\sigma) = \mathit{mr}(c;\tau)$ and $\mathit{mr}(c';\sigma) < \mathit{mr}(c';\tau)$
\end{proof}
  As a result of \ref{bruhat-locality} if $r_i(\mathcal{T})$ is the number of elements in $\mathcal{T}$ with reuse distance $i$ or smaller, we have 
\[
\sum_{i=1}^{m-1} r_i(\mathcal{T}) = \ell(\sigma)
\]
This statement is equivalent to Theorem \ref{bruhat-locality}. Note that since this sum truncates the highest possible reuse distance ($m$), a higher value is better for locality, as more elements with $r_i < m$ implies less elements with $r_i = m$.  Then, since $\ell$ structures the partial order of $S_m$ following the Bruhat order, we have a direct relation between the ordering imposed by $\ell$ and locality. That is, $\ell(\sigma) > \ell(\tau)$ implies that $\sigma$ has better temporal locality than $\tau$. This result and the theorems that lead to it are the basis of our symmetric theory of locality.

\begin{theorem}\label{opt-relabe-theorem}
    If $\sigma$ is the optimal reordering for $A$ such that $A\,\sigma(A)$ has better locality than $A\,\tau(A)$ for $\tau \neq \sigma$, then  $\sigma(A)\,A$ has better locality than $\sigma(A)\,\tau(A) $ for $\tau \neq e$. 
\end{theorem}
\begin{proof}
    Since Reuse Distance can be calculated backwards or forwards equivalently, $\sigma(A)\,A$ has the same temporal locality as $A\,\sigma(A)$, which is assumed to be optimal. Assume there was some other $\tau$ such that $\sigma(A)\,\tau(A)$ has better locality than $\sigma(A) \,A$. Then, by a relabeling argument, 
    \[\sigma(A)\,\tau(A) \stackrel{\mathclap{\tiny\mbox{locality}}}{\simeq} A\, \left(\left(\sigma \tau\right)(A)\right),\] 
    so $\sigma \tau $ has better locality than $\sigma$, which is a contradiction.
\end{proof}

\subsection{Experimental Results}
To support our symmetric theory of locality, we aggregate average cache miss ratio curves for each inversion number a permutation of a symmetric group can have in \ref{mrc}. To do this, we consider an element-wise average for each cache size. 
\begin{figure}
    \raggedleft
    \includegraphics[width=\columnwidth]{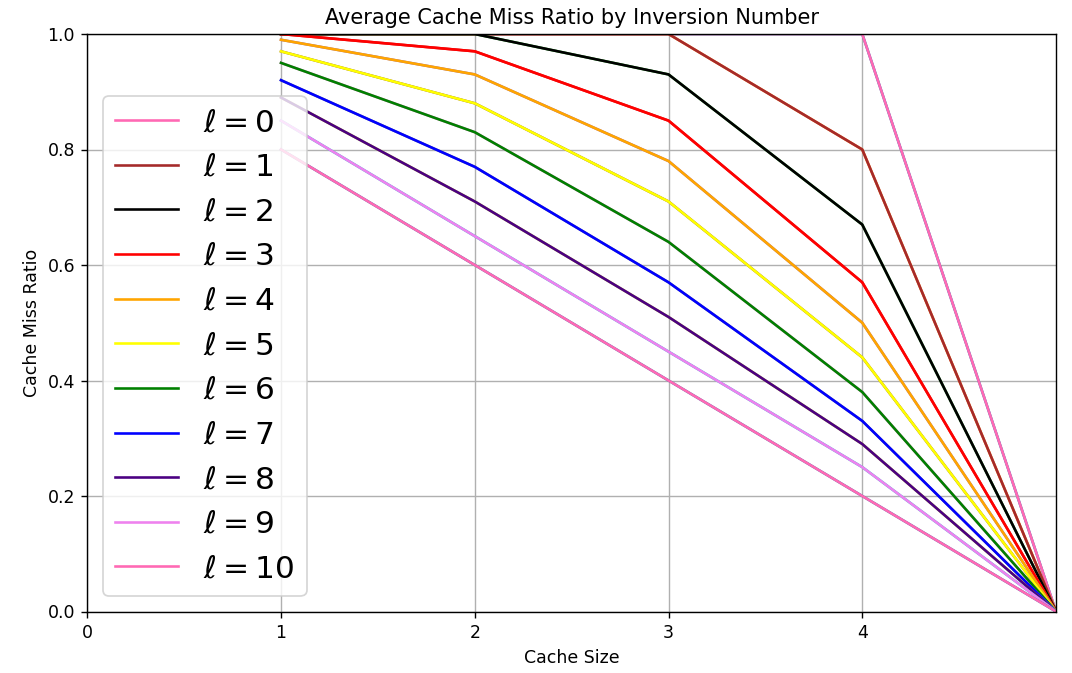}
    \caption{Average Miss Ratio Curve by Inversion Number for $S_5$}\label{mrc}
    \label{fig:avgs}
\end{figure}

We see that this depicts a clear trend and separation following the Bruhat ordering for symmetric groups, supporting the earlier theorems that prove a connection between the Bruhat ordering and locality. We also notice a decrease in convexity as $\ell$ approaches the maximum size. The trends continue for larger graph sizes, however graphs become hard to read for large cache and so we include only cache size up to 5.


\section{ChainFind Algorithm}\label{section:chainfind}
\subsection{Chains $\&$ Labelings}

The ChainFind algorithm is inspired by the construction of poset complexes from algebraic topology. \cite{wachs2006poset} Let $H = (S_m, \bruhatlhd)$ be a covering system of $S_m$ w.r.t. the Bruhat order. Let $\lambda$ be a EL-labeling edge labeler. The following algorithm is a greedy weighted breadth first search: given a current path, it finds the best possible path among feasible edges w.r.t. $Q$, and iteratively builds this path until some stopping point. The full algorithm in its entirety is described in algorithm \ref{alg:chainfinding}. The stack $P$ returned is the chain starting from $\tau_0$ to some $\tau_p$, where $p = |P|, \tau_p \leq \tau_{\mathit{rev}}$. The next focus is to design good edge labelings of $\lambda$ and study their consequences or variations.

With respect to the Bruhat order, all maximal chains in $S_m$ are of length $O(m^2)$. \cite{Abello2004} Algorithm \ref{alg:chainfinding} gives the best possible locality for a chain w.r.t. some totally ordered $Q$ edge-labels, where each edge label is calculated from the locality of the edge's destination node. The covering relation ensures that the branching explored by the algorithm is limited by at most, $|T| = O(m),$ where $T$ is the set of reflections of $S_m$. Therefore, the total runtime is $O(m^3)$.

The classification of the ChainFind algorithm as a greedy algorithm is proven in Appendix \ref{appendix:el-labeling}. In the context of compiler design, where performance is critical, classifying the algorithm as greedy highlights its efficiency in making locally optimal decisions at each step.

\begin{algorithm}
\caption{Chain Finding Algorithm}\label{alg:chainfinding}
\begin{algorithmic}
\Require $m$ is the number of data access elements
\Ensure $\tau_0 \in S_m, \tau_0 \leq \tau_{\mathit{rev}}, \tau_{\mathit{rev}} \text{ is the reverse permutation.}$
\State $\mathit{chain} \gets \mathit{Stack}(\tau_0)$
\While{$\mathit{size}\left(\mathit{chain}\right) < \frac{m(m + 1)}{2}$}
    \State $p_c \gets \mathit{peek}(\mathit{chain})$
    \State $\mathit{E}_{\mathit{size}} = \{y\,:\, p_c \lhd y\}$
    \State $next \gets max(\mathit{E}_{\mathit{size}})$
    \State $push(chain, next)$
\EndWhile
\State \Return $\mathit{chain}$

\end{algorithmic}
\end{algorithm}

\IEEEPARstart{I}{n} practice, finding chains should only be done for memory-loading intensive algorithms, For the proposed ChainFind algorithm, a running time of $O(m^3)$ may be too slow in cases that do not take advantage of intensive memory accesses. 
If the algorithm is being run multiple times, the chain finding algorithm can be done as part of a \emph{Just-in-time} (JIT) scheme by caching (parts of) the chain during profiling, which massively reduces overhead. \cite{jvmjit}

\subsection{Locality Ordering}

\noindent We propose several new and heuristic orderings that describe the locality of traces under permutation. For the purposes of this section, let $H$ represent the graph $(S_m, \bruhatlhd)$. Recall that $\mathit{hits}_c(\sigma)$ measures the locality of $\sigma$ with cache size $c$.

Before we delve into a more in-depth study on possible locality orderings, we must discuss feasibility. We say that some trace is \emph{feasible} if it is a possible trace that can be generated by a program. Due to topological dependencies and the AST structure of programs, not every trace is feasible for a given program. Thus, when computing the permutation of traces, we must make sure to stay within the feasible space.
\begin{definition}[Feasibility]
    An access trace is \textbf{feasible} if it can be generated from a given program $P$. Let $\mathcal{F}(P)$ be the space of feasible traces. \emph{Infeasibilities} may occur when there exist topological dependencies within a program that create restrictions on which elements (or edges in our graph) we can access while maintaining program correctness. In this case we focus our attention on a subset of the graph of $S_m$.
\end{definition}
In the context of making a labeler $\lambda$, we may define a binary \emph{feasible} function $Y$:
\[
Y(\mathcal{T}) = \begin{cases}
    1 & \mathcal{T} \in \mathcal{F}(P) \\
    0 & \text{ otherwise}
\end{cases}
\]
In terms of the lexicographically ordering, the chain generated should yield traces that are feasible. However, it is not necessarily \emph{maximal}, but in the context of program optimization, that does not matter.

For mathematical compatibility, we assume henceforth that every possible trace is feasible.

\subsubsection{Miss Ratio Labeling}
A naive way to design the edge labeler is to simply consider the lexicographic ordering of $\mathit{hits}_C$ (and likewise, the corresponding miss ratio curve $\mathit{MRC}$). 
\[\lambda(\sigma, \tau) = \mathit{hits}_C(\tau) = (\mathit{hits}_1(\tau), \mathit{hits}_2(\tau), \hdots, \mathit{hits}_m(\tau))
\]
The ChainFind algorithm will make a decision with comparing the number of cache access hits for $c = 1$ among all $\tau$. If the feasible traces are the same, then it will test $c = 2, 3, \hdots, m$. However, it is not a good labeling. Consider the counterexample at $\sigma = e$, and thus any $\sigma_i \in \mathcal{S}$ implies $e \lhd \sigma_i$. With respect to $c = 1$, we have $\mathit{hits}_1(s_i) = 0$, so $\lambda(e, s_i) \neq \lambda(e,s_j)$ for $i \neq j$.

One option to deal with this is to simply ignore it by creating an arbitrary tiebreaker, since as according to the current labeling, each of the feasible traces have equivalent locality. For compatibility with the good labeling property, the tiebreaker will have to be designed in a way to take into account the total ordering. Examples include a random tiebreaker and the $\sigma_i$ that described the edge. In particular, the usage of $\sigma_i$ is inspired by the \emph{standard labeling} of the $S_m$ as a Coexter group. \cite{phdthesis} 

\subsubsection{Ranked Miss Ratio Labeling}
Another option is to permute the miss ratio curve, which ranks the importance of cache size. Often, cache systems are hierarchical based on cache size, which also makes this option attractive. For $\psi \in S_m$, we define
    \begin{align*}
    \lambda^\psi(\sigma, \tau) &= \mathit{hits}_{\psi(C)}(\tau) \\
    &= \left(\mathit{hits}_{\psi(1)}(\tau), \mathit{hits}_{\psi(2)}(\tau), \hdots, \mathit{hits}_{\psi(m)}(\tau)\right).  
\end{align*}
    
Different cache sizes are preferred w.r.t. $\phi$. For example, if we let $\psi(1) = m - 1$, then the counterexample previously mentioned does not need a tiebreaker. However, in practice, this does not alleviate the problem; we give an example of generating a chain from $S_{11}$, with $\psi = (1\;10\;9\;8\;7\;6\;5\;4\;3\;2)$, which essentially slides $\mathit{hits}_{c = 10}$ value in front of the $\mathit{hits}_{C = 11}$ vector. The total length of the chain is $66$, but there is a factor of $9$ different chains that could be made, as compared to $\psi = e$, where there is a factor of $14$ different chains. Figure \ref{fig:chain_choice} shows that as the $m$ from $S_m$ increases, the factor of chains that could be made also roughly increases, which means that the chain generated from the algorithm is not distinct with respect to $\lambda^e$ (as well as any $\phi \neq e$).

\begin{figure}
    \raggedright
    \includegraphics[width=\columnwidth]{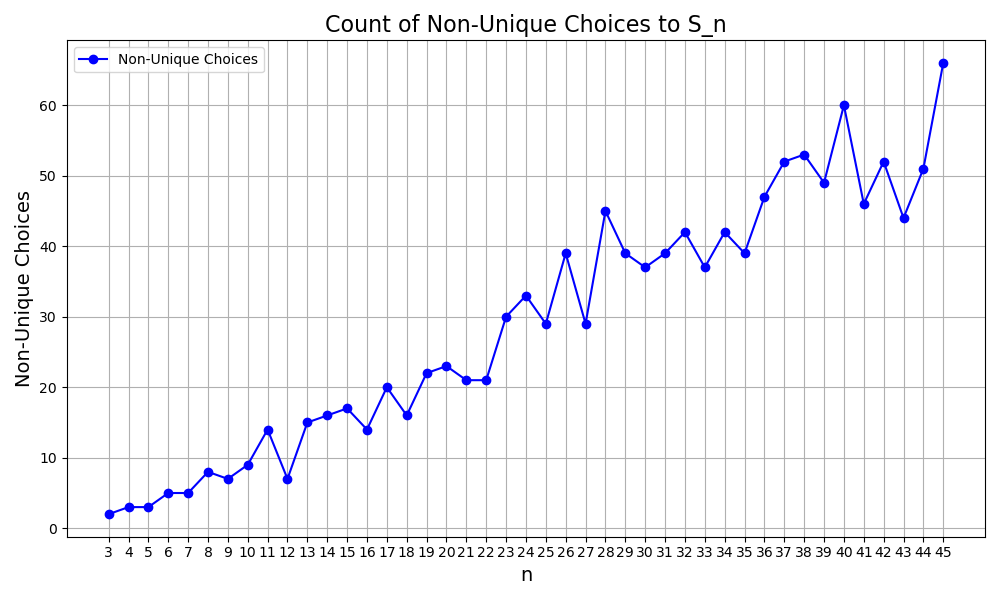}
    \caption{Plot of $\lambda^e$: $S_n$ vs. the count of arbitrary choices that ChainFind had to decide.}
    \label{fig:chain_choice}
\end{figure}

\section{Discussion}\label{discussion}
We present the following \textbf{open problem}:
\begin{problem}\label{problem-3}
    Does there exist an EL-labeling (\ref{el-label}) $\lambda$ that is dependent  \emph{precisely} on locality? In other words, is the \emph{ChainFind} algorithm that depends on this $\lambda$ both optimal and computationally efficient?
\end{problem}
To our knowledge, we do not know if it is possible to design such a $\lambda$. In order to find a \emph{good labeling}, we attempted to propose several other orderings, where the most notable were \emph{timescale locality} \cite{relational-locality} and \emph{Data Movement Complexity} \cite{Smith+:ICS22}, among others. In order to advance this problem, we suggest first coming up with a proper \emph{good labeling}, and in turn, transforming to a full \emph{EL-labeling}.

The solution to problem \ref{problem-3} yields a particularly intriguing insight into the nature of program optimization concerning locality: a valid solution would suggest that the optimal strategy is greedy. This finding could open up new avenues for exploring the interplay between greedy algorithms and locality-aware optimization in caching systems.

It is also unclear if the construction of an EL-labeling does matter; maybe a good labeling suffices for most program optimizations. 

For the rest of the section, we discuss several areas of exploration to take advantage of symmetric locality. In addition to the below sections, we would like to add that a positive effect of our \emph{ChainFind} algorithm always takes program correctness into account with the addition of the feasibility boolean function $Y.$ 

\subsection{Application To Deep Learning}

\subsubsection{Permutation Equivariance}
\noindent Recently, permutation equivariant deep learning models have been developed for their flexibility, ability to address privacy \cite{equitrans}, and handle weight-space tasks \cite{penf}.  We define a function $f$ to be permutation equivariant if, for some $\sigma \in S_n$:
\begin{equation*}
    \sigma f(x) = f(\sigma x)
\end{equation*}
Functions found to hold this property include element-wise operators, softmax, linear layers, normalization layers, attention, and others \cite{equitrans}.  Since backward propogation is also permutation invariant, we can apply permutations to full models, including MLPs, CNNs, transformers \cite{equitrans}, and GNNs~\cite{egnn}.

\IEEEPARstart{H}{owever}, this concept does not apply to all data types. The data must also be permutation invariant so it does not lose meaning when permuted.  We define the notion of "order" for data to be reflect whether you can permute elements of data without information loss. Some data is completely unordered, like a set of stock prices, while other data is completely ordered, such as a novel. We will refer to unordered data as a set. We also encounter notions of partially ordered data, where the relative order of some elements must hold, but variance is permitted. Examples include a series of sentences where the 
sentences can be permuted but words within can not,
or randomly sampled movements of particles over time where the order of the time stamps matters but the order of particles within the time stamps do not.

\subsubsection{Optimizing Locality}
\noindent We propose how to optimize locality for different orders of data. We investigate the case of linear layers (viewed as matrices) for MLPs, and assume a permutation equivariant activation function. We establish an essential theorem to assess repeated accesses of a sequence $A$.
As a result of theorem \ref{opt-relabe-theorem}, we can see that if $\sigma(A)$ is the optimal reordering of $A$, then the optimal reordering of accessing a sequence $A\,A\,A\,A\,A\dots$ is $A\,\sigma(A)\,A\,\sigma(A)\,A\dots$.
Applying this to learning models, first time we encounter a linear layer, we compute without permutation. However, the second time we access a weight (in backpropogation), we access the tensor with regards to the optimal permutation while calculating loss. Then, upon seeing the weight again we resume initial order. Similarly, this concept can be applied to the key, value, and projection matrices in multihead-attention to greatly improve the memory efficiency of transformers. This can be applied to both training and deployment, as this optimization permits but does not necessitate backpropagation. \\ To illustrate the benefit of this, we provide a comparison of reuse distances for repeated $n \times m$ matrix (input to MLP) accessed in cyclic and sawtooth  re-traversal order. In cyclic order, we see $nm$ elements each with reuse distance $nm$, yielding a total reuse of $n^2m^2$. However, a sawtooth traversal would yield $\sum_{i=1}^{nm} i = \frac{nm(nm+1)}{2}$ reuse distance. The leading term is halved, leading to a significant improvement of temporal locality. 

\IEEEPARstart{F}{or} unordered sets of data, the sawtooth permutation should be applied to optimize locality. Considering partially ordered data, we refer to the covering graph to inform our decision. The highest access order on the covering graph that preserves the partial order of the data should be selected to optimize locality. For completely ordered sets, we cannot permute data, so our optimization does not apply.

\subsection{Instruction Locality}
We proposed that this new theory can be applied under the guise of a compiler or JIT optimization. It would reorder program instructions when they are executed repeatedly.  In the context of very simple programs, it is immediately applicable.   Instruction scheduling, however, serves other purposes such as instruction level parallelism and register reuse.  Also, typical loop bodies fit in instruction cache and may not benefit further from the reordering optimization.

\subsection{Reordering Algorithms}
Reordering algorithms, such as those used in preprocessing for GNNs \cite{reorder}, intend on optimizing ordering to maximize spatial and temporal locality by relabeling the nodes of graphs. Our analysis focuses on temporal locality, and could prove to be a useful tool to improve reordering efforts for subsets of graphs that undergo repeated traversals, such optimizing algorithms traversing as a set of vertices that share many neighbors. 

\subsection{Non-periodic Data Reuse}
For the periodic trace $\mathcal{T} = A \, B$, each data is reused at most once. This is a problem as in real world caches, data is reused any arbitrary amount of times. To support this new theory, we would have to compare general traces to periodic traces.

\subsection{Limitations}
\subsubsection{Non-periodic Data Reuse}
One limitation is that the theory targets only data re-traversals. Modeling by permutations does not cover the extent of all real world cache accesses.  
The theory models fully associative LRU caches only and does not consider the effect of parallelism.  
While memory performance can greatly benefit from latency hiding, e.g., prefetching, we focus on locality and consider only the data movement, not the running time.

\section{Final Remarks}\label{conclusion}
We establish a theory of symmetric locality\ref{symmlocaldef}, comparing data re-traversals and providing guidelines for optimization opportunities. Fundamentals of group theory, in particular the Bruhat order, are outlined to ground our work. We endow a concept of good labeling, and provide a chain finding algorithm for efficiently finding good labelings alongside an algorithm for efficiently computing reuse distance. Symmetric locality is particularly useful for programs with code or data dependencies, in which sawtooth traversal may not be applicable.

The theory of symmetric locality proposed is not complete. We outline possible future research directions.  To motivate a usage to implement in applications in compiler designs or systems, the parallelism or scheduling problem is a very important issue to tackle. Not all programs have data reuse of its accesses being at most one, so the theory may not applicable for all programs. Expanding the applicable class size and addressing multi-threading will allow for a more comprehensive theory of symmetric locality.

\section*{Acknowledgments}
We would like to thank the following for proofreading this paper: Jack Cashman, Leo Sciortino, Willow Veytsman, Woody Wu \& Yiyang Wang. We wish to thank Donovan Snyder for his alternative proof of Theorem \ref{bruhat-locality}.  




\bibliographystyle{IEEEtran}
\bibliography{refs}

\section{Appendix}\label{appendix}
\subsection{Groups} \label{groups}
\begin{definition}
    A group $\mathcal{G}$ is a set of elements with a binary operator $\ast$ satisfying the following relations: 
    \begin{itemize}
        \item \emph{Closure}: $\forall a,b \in \mathcal{G}, a \ast b \in \mathcal{G}$
        \item \emph{Associativity}: $\forall a,b,c \in \mathcal{G}$, $(a \ast b) \ast c = a \ast (b \ast c)$
        \item \emph{Identity}: $\exists ! e \in \mathcal{G}$ such that $\forall a \in \mathcal{G}$, $e \ast a = a$ and $a \ast e = a$.
        \item  \emph{Invertibility}: $\forall a \in \mathcal{G}$, $\exists ! b \in \mathcal{G}$ such that $a \ast b = e$ and $b \ast a = e$.
    \end{itemize}
\end{definition}

\begin{definition}
    A subset $\mathcal{S}$ of elements of a group $\mathcal{G}$ is called a generator if all elements in $\mathcal{G}$ can be expressed as finitely many combinations of elements in $\mathcal{S}$ and the inverses of elements in $\mathcal{S}$.
\end{definition}

\begin{definition}
    A set $\mathcal{R}$ of relations is a set of constraints placed on elements. For example, a constraint could be that for elements $a, b \in \mathcal{G}$, $a \ast b = b \ast a$.
\end{definition}

\begin{definition}
    A presentation $\braket{\mathcal{S}|\mathcal{R}}$ is a set of generators and rules that generate a group. This means that to describe a group, all we would need are the set of its generators and its relations.
\end{definition}

\subsection{Symmetric Group} \label{symmetric_group}

\begin{definition}
    A permutation is a bijective mapping $\sigma: A \to A$ that often serves to rearrange elements of $A$. 
\end{definition}

\begin{definition}
    A symmetric group $S_m$ defined over a set of $m$ objects is a group whose elements are permutations of these objects, and whose operator is function composition.
\end{definition}

\begin{definition}
    We can describe any permutation as a composition of k-cycles, which is a cycle with $k$ elements. For example, (1,2, 3) is a 3-cycle, while (1,3)(2,4) is 2 2-cycles. We say that any cycle of 2 elements $(a,b)$ is called a \textbf{transposition}.
\end{definition}

\begin{lemma} \label{cycledecomp}
    \emph{Cycle Decomposition Theorem}:  Any k-cycle $(a_1 \dots a_k)$ can be decomposed into not necessarily distinct series of transpositions: \[(a_1 \dots a_k)  = (a_1 a_k) (a_1 a_{k-1} ) \dots (a_1 a_2).\] 
\end{lemma}While this decomposition is not unique, its parity is.
\begin{definition}
    For the symmetric group $S_m$, let $A$ be the set of \textbf{adjacent transpositions}, also known as swaps: \[A = \{(i, i + 1) : \forall i \in [m - 1]\}.\] For $i \in [0, m - 1]$ we will refer to $\sigma_i$ as the adjacent transposition $(i, i+1)$. $A$ is also the set of \textbf{generators} of $S_m$; we can obtain any permutation by composing any finite number of swaps.
\end{definition}

\begin{definition}
    Let $T$ represent the set of reflections of $S_m$:
    \begin{equation*}
        T = \{\tau \sigma_i \tau^{-1} : \sigma_i \in A, \tau \in S_m \}
    \end{equation*}
    $T$ describes swaps between any two elements.
\end{definition}


\begin{lemma}
    The symmetric group $S_m$ is a Coxeter group with the set of reflections $T$ \cite{phdthesis}. The proof starts with first noticing that $S_m$ is a presentation $\braket{\mathcal{S}|\mathcal{R}}$, such that
    \begin{align*}
        \mathcal{S} &= A \\
        \mathcal{R} &= \begin{cases}
            \sigma_i \sigma_j          = \sigma_j \sigma_i & |i - j| \geq 2 \\
            \sigma_i \sigma_j \sigma_i = \sigma_j \sigma_i \sigma_j & |i - j| = 1
        \end{cases}
    \end{align*}
    where $A$ is the set of adjacent swaps. \cite{phdthesis}
\end{lemma}


\subsection{Partial Order}
The following is also covered in \cite{phdthesis}. 
\begin{definition}
Let $(P, \leq)$ be a set with a partial ordering $\leq$, that is $\forall x, y, z \in P$:
\begin{itemize}
    \item \emph{Reflexivity}: $x \leq x$
    \item \emph{Transitivity}: If $x \leq y$ and $y \leq z$ implies $x \leq z$
    \item \emph{Anti-symmetry}: If $x \leq y$ and $y \leq x$ implies $x = y$
\end{itemize}
\end{definition}

Note, this doesn't imply that every single element is comparable. If this is the case, then the partial ordering becomes a total ordering.

If $P$ is a nonempty set, then we say that $\alpha = \sup P$ is the supremum of P if $x \leq \alpha, \forall x \in P$. Likewise, we say that $\beta = \inf P$ is the infirmum of P if $\beta \leq x, \forall x \in P$. If both of these exist, then P is bounded.  

\subsection{Covering \& Order}\label{appendix:covering}
We seek to establish a covering system and order on the symmetric group $S_m$ in order to institute structure for analysis. This yields a natural graph representation of $S_m$, which we wish to traverse while maintaining optimal locality at each step. This section establishes necessary building blocks to achieve these goals.
\begin{definition}
    Let $\sigma, \tau \in S_m$, and let $T \subset S_m$ be the set of all transpositions (swaps). Then if $\tau  = s_1 s_2 \hdots s_p, \: s.t.\,\: s_j \in T$, we define an operator $\bruhatleq$ below:
    \[\sigma \bruhatleq \tau \iff \sigma = s_{i_1} s_{i_2} \hdots s_{i_q}\]
    where $i_1, i_2, \hdots, i_q$ is a subsequence of $1, 2, \hdots, p$. The operator $\bruhatleq$ is defined as the \textbf{Bruhat Order}.
\end{definition}

\begin{definition}
    Let $H = (V, E)$ be a directed graph of $S_m$ wrt to $\leq$, where $V = S_m$, and 
    \[ E = \{(\sigma, \tau)\,:\, \ell{(\tau)} = \ell{(\sigma)} + 1 \land \sigma \bruhatleq \tau\}\]
    In other contexts, $H$ is also called a \emph{covering system}. Since $\ell$ precisely constructs the edges, we also equivalently write $H = (S_m, \ell)$. By \ref{invlemma}, the edge also corresponds to the $s_i \in T$ such that $\tau = \sigma s_i$.
\end{definition}
\emph{Example}: Let $\sigma = (13)$ and $\tau = (14)(13)$. $\ell(\sigma) = 3$ and $\ell(\tau) = 4$ by definition of inversion. In terms of the cycle decomposition, we can see that 
\begin{align*}
    \sigma &= (12)(23)(12) \\
    \tau   &= (12)(23)(34)(23)(12)(12)(23)(12) \\
           &= (12)(23)(34)(23)(23)(12) \\
           &= (12)(23)(12)(34)
\end{align*}
The decomposition of $\sigma$ is a subsequence of $\tau$, so $\sigma \leq \tau$. Also, the above elements make an edge in $H$, by inspecting the decomposition length $\ell.$

To clarify notation, we also denote
\begin{align*}
    x \bruhatlq y &\iff x \bruhatleq y \land x \neq y \\
    x \bruhatlhd y &\iff x \bruhatlq y \land \ell{(y)} = \ell{(x)} + 1
\end{align*}

\subsection{Chains \& EL-Labeling}\label{appendix:el-labeling}
\begin{definition}
    A \textbf{chain} is a totally ordered collection of elements from $S_m$. A \textbf{saturated chain} is a chain with maximal possible length.
\end{definition}
Let $Q$ be a totally ordered set. We use the edge labeler $\lambda$ below to describe the process of labeling edges with element of $Q$:
\[ \lambda: \{(x, y) \,:\, x \bruhatlhd y\} \to Q \]
\begin{definition}\label{el-label}
    Let $\lambda$ be a edge labeling function. It is an \textbf{EL-labeling} (Edge Lexicographic labeling)\footnote{This property leads to another "nice" property known as EL-shellability.}, if $\forall x, y \in S_m, x < y$:
    \begin{enumerate}
        \item \label{el-label-1} There is exactly one saturated chain from x to y such that the labels are in a non-decreasing order.
        \item \label{el-label-2} For any 2 saturated chains, the labels of one chain is smaller than the other w.r.t. the dictionary order. 
    \end{enumerate}
    A chain with (2) w.r.t to any other is called a \textbf{minimal}\footnote{Maximal refers to the length of the chain. Minimal refers to its labels.} chain. 
\end{definition}
We would like to point out that the construction of an \emph{EL-labeling} directly implies the optimality \ref{el-label}.\ref{el-label-1} and efficiency \ref{el-label}.\ref{el-label-2} of a greedy algorithm, we call this algorithm the \emph{ChainFind} algorithm.
\begin{definition}
    Let $Q$ be a totally ordered set (of labels). For a successor system $H$ of $S_m$, we say that $\lambda: H \to Q$ is a \emph{good labeling} if \[\forall x, y, z \in H, x \bruhatlhd y \land x \bruhatlhd z,\; \lambda(x, y) = \lambda(x, z) \implies y = z\]
    In other words, the multiple choices that increase the length of a chain are distinct wrt $\lambda$ (or what is more commonly known as a bijective property).
\end{definition}
\begin{lemma}
    If $\lambda$ is a good labeling, then \ref{el-label}.\ref{el-label-2} is fulfilled. \cite{phdthesis}
\end{lemma}
A good labeling ensures that the chain finding algorithm retains its $O(m^3)$ running time. If a good labeling also satisfies $\ref{el-label-1}$, then the generated chain is optimal in terms of the $\lambda$ defined.


\subsection{Miscellaneous Characterizations} \label{misc}
The following observations provide additional facts obtained from working with re-traversals.\\
We can characterize each level by noting that possible cache hit vectors are all possible integer partitions for $\ell = n$. Furthermore, if we add all occurrences of each integer partition for a rank of the symmetric group, we find the Mahonian number $M(m,n)$, which counts how many permutations of $m$ elements have $n$ inversions. How to count the number of cache vectors with a specific integer partition is an open problem. Furthermore, if we were to integrate  the normalized truncated cache miss vector (the normalized cache hit vector without the last element), we would note that vectos with the same inversion number evaluate to the same value, and the value of integrals drops from 1 (at the identity) to .5 (at sawtooth) with slop $\frac{1}{m(m-1)}$.

\end{document}